\documentclass{lmcs} 


\usepackage{hyperref}

\usepackage[linesnumbered,ruled]{algorithm2e} 
\usepackage{amsmath}
\usepackage{amssymb}
\usepackage{centernot}
\usepackage{eepic}
\usepackage{stmaryrd}
\usepackage{tikz}
\usetikzlibrary{automata, positioning, arrows}
\tikzset{
  ->, 
  >=stealth', 
  node distance=3cm, 
  every state/.style={thick, fill=gray!10}, 
  initial text=$ $, 
}
\usepackage{xspace}

\newcommand{\derives}[1]{\xrightarrow{#1}}
\newcommand{\lang}{\mathcal{L}}
\newcommand{\out}{\textit{out}}
\newcommand{\rexp}{\mathcal{R}}
\newcommand{\setof}[2]{\{ #1 \mid #2 \}}

\newcommand{\boxit}[1]{\begin{array}{|c|} \hline\\[-2.2ex] #1 \\\hline \end{array}} 
\newcommand{\postrule}[3]{
			\boxit{\begin{array}{c} #1 \\\hline\\[-2.2ex] #2 \end{array} \;#3}}

\theoremstyle{plain} 

\begin{document}

\title[Better Automata through Process Algebra]{Better Automata through Process Algebra}

\author[R.~Cleaveland]{Rance Cleaveland}	
\address{Department of Computer Science, University of Maryland, College Park MD 20742 USA}	
\email{rance@cs.umd.edu}  
\thanks{Research supported by US Office of Naval Research Grant N000141712622.}	

%





\begin{abstract}
  \noindent This paper shows how the use of Structural Operational Semantics (SOS) in the style popularized by the process-algebra community can lead to a more succinct and useful construction for building finite automata from regular expressions.  Such constructions have been known for decades, and form the basis for the proofs of one direction of Kleene's Theorem.  The purpose of the new construction is, on the one hand, to show students how small automata can be constructed, without the need for empty transitions, and on the other hand to show how the construction method admits closure proofs of regular languages with respect to other operators as well.  These results, while not theoretically surprising, point to an additional influence of process-algebraic research:  in addition to providing fundamental insights into the nature of concurrent computation, it also sheds new light on old, well-known constructions in automata theory.
\end{abstract}

\maketitle

\section{Introduction}\label{S:one}

It is an honor to write this paper in celebration of Jos Baeten on the occasion of the publication of his \emph{Festschrift}.  I recall first becoming aware of Jos late in my PhD studies at Cornell University.  Early in my doctoral career I had become independently interested in process algebra, primarily through Robin Milner's original monograph, \emph{A Calculus of Communicating Systems}~\cite{Milner80}, and indeed wound up writing my dissertation on the topic.  I was working largely on my own; apart from very stimulating interactions with Prakash Panangaden, who was at Cornell at the time, there were no researchers in the area at Cornell.  It was in this milieu that I stumbled across the seminal papers by Jos' colleagues, Jan Bergstra and Jan Willem Klop, describing the Algebra of Communicating Processes~\cite{BERGSTRA1984109,bergstra1985algebra}.   I was impressed with their classically algebraic approach, and their semantic accounts based on graph constructions.  This, together with Milner's focus on operational semantics and the Communicating Sequential Processes community's on denotational semantics~\cite{BrookesHR84}, finally enabled me to truly understand the deep and satisfying links between operational, denotational and axiomatic approaches to not only process algebra, but to program semantics in general.

While Jos was not a co-author of the two papers just cited, he was an early contributor to the process-algebraic field and has remained a prolific researcher in both theoretical and applied aspects of the discipline. I have followed his career, and admired his interest in both foundational theory and practical applications of process theory, since completing my PhD in 1987.  It is this broader view on the impact of process algebra that is the motivation for this note.  Indeed, I will not focus so much on new theoretical results, satisfying though they can be.  Rather, I want recount a story about my usage of process-algebra-inspired techniques to redevelop part of an undergraduate course on automata theory that I taught for a number of years.  Specifically, I will discuss how I have used the Structural Operational Semantics (SOS) techniques used extensively in process algebra to present what I have found to be more satisfying ways than those typically covered in textbooks to construct finite automata from regular expressions.  Such constructions constitute a proof of one half of Kleene's Theorem~\cite{kleene1956}, which asserts a correspondence between regular languages and those accepted by finite automata.

In the rest of this paper I present the construction and contrast it to the constructions found in classical automata-theory textbooks such as~\cite{hopcroftMU2006}, explaining why I find the work presented here preferable from a pedagogical point of view.  I also briefly situate the work in the setting of an efficient technique~\cite{berry1986regular} used in practice for converting regular expressions to finite automata.  The messsage I hope to convey is that in addition to contributing foundational understanding to notions of concurrent computation, process algebra can also cast new light on well-understood automaton constructions as well, and that pioneers in process algebra, such as Jos Baeten, are doubly deserving of the accolades they receive from the research community.


\section{Alphabets, Languages, Regular Expressions and Automata}

This section reviews the definitions and notation used later in this note for formal languages, regular expressions and finite automata.  In the interest of succinctness the definitions depart slightly from those found in automata-theory textbooks, although notationally I try to follow the conventions used in those books.

\subsection{Alphabets and Languages}

At their most foundational level digital computers are devices for computing with symbols.  Alphabets and languages formalize this intuition mathematically.

\begin{defi}[Alphabet, word]
\hfill
\begin{enumerate}
\item
An \emph{alphabet} is a finite non-empty set $\Sigma$ of symbols.
\item
A \emph{word} over alphabet $\Sigma$ is a finite sequence $a_1\ldots a_k$ of elements from $\Sigma$.  We say that $k$ is the \emph{length} of $w$ in this case.  If $k = 0$ we say $w$ is \emph{empty}; we write $\varepsilon$ for the (unique) empty word over $\Sigma$.  Note that every $a \in \Sigma$ is also a (length-one) word over $\Sigma$.  We write $\Sigma^*$ for the set of all words over $\Sigma$.
\item
If $w_1 = a_1 \ldots a_k$ and $w_2 = b_1 \ldots b_\ell$ are words over $\Sigma$ then the \emph{concatenation}, $w_1 \cdot w_2$, of $w_1$ and $w_2$ is the word $a_1 \ldots a_k b_1 \ldots b_n$.  Note that $w \cdot \varepsilon = \varepsilon \cdot w = w$ for any word $w$.  We often omit $\cdot$ and write $w_1w_2$ for the concatenation of $w_1$ and $w_2$.
\item
A \emph{language} $L$ over alphabet $\Sigma$ is a subset of $\Sigma^*$.  The set of all languages over $\Sigma$ is the set of all subsets of $\Sigma^*$, and is written $2^{\Sigma^*}$ following standard mathematical conventions.
\end{enumerate}
\end{defi}

Since languages over $\Sigma^*$ are sets, general set-theoretic operations, including $\cup$ (union), $\cap$ (intersection) and $-$ (set difference) may be applied to them.  Other, language-specific operations may also be defined.

\begin{defi}[Language concatenation, Kleene closure]
Let $\Sigma$ be an alphabet.
\begin{enumerate}
\item
Let $L_1, L_2 \subseteq \Sigma^*$ be languages over $\Sigma$.  Then the \emph{concentation}, $L_1 \cdot L_2$, of $L_1$ and $L_2$ is defined as follows.
\[
L_1 \cdot L_2 = \{ w_1 \cdot w_2 \mid w_1 \in L_1 \textnormal{ and } w_2 \in L_2 \}
\]
\item
Let $L \subseteq \Sigma^*$ be a language over $\Sigma$.  Then the \emph{Kleene closure}, $L^*$, of $L$ is defined inductively as follows.\footnote{Textbooks typically define $L^*$ differently, by first introducing $L^i$ for $i \geq 0$ and then taking $L^* = \bigcup_{i=0}^\infty L^i$}
\begin{itemize}
\item $\varepsilon \in L^*$
\item If $w_1 \in L$ and $w_2 \in L^*$ then $w_1 \cdot w_2 \in L^*$.
\end{itemize}
\end{enumerate}
\end{defi}

\subsection{Regular Expressions}

\emph{Regular expressions} provide a notation for defining languages.

\begin{defi}[Regular expression]
Let $\Sigma$ be an alphabet.  Then the set, $\rexp(\Sigma)$, of \emph{regular expressions} over $\Sigma$ is defined inductively as follows.
\begin{itemize}
\item $\emptyset \in \rexp(\Sigma)$.
\item $\varepsilon \in \rexp(\Sigma)$.
\item If $a \in \Sigma$ then $a \in \rexp(\Sigma)$.
\item If $r_1 \in \rexp(\Sigma)$ and $r_2 \in \rexp(\Sigma)$ then $r_1 + r_2 \in \rexp(\Sigma)$ and $r_1 \cdot r_2 \in \rexp(\Sigma)$.
\item If $r \in \rexp(\Sigma)$ then $r^* \in \rexp(\Sigma)$.
\end{itemize}
\end{defi}

It should be noted that $\rexp(\Sigma)$ is a set of expressions; the occurrences of $\emptyset, \varepsilon, +, \cdot$ and $^*$ are symbols that do not innately possess any meaning, but must instead be given a semantics.
This is done by interpreting regular expressions mathematically as languages.  The formal definition takes the form of a function, $\lang \in \rexp(\Sigma) \rightarrow 2^{\Sigma^*}$ assigning a language $\lang(r) \subseteq \Sigma^*$ to regular expression $r$.

\begin{defi}[Language of a regular expression, regular language]~\label{def:rexp-language}
Let $\Sigma$ be an alphabet, and $r \in \rexp(\Sigma)$ a regular expression over $\Sigma$.  Then the \emph{language}, $\lang(r) \subseteq \Sigma^*$, associated with $r$ is defined inductively as follows.
\[
\lang(r) =
\left\{
\begin{array}{lp{5cm}}
\emptyset
	& if $r =\emptyset$
	\\
\{ \varepsilon \}
	& if $r = \varepsilon$
	\\
\{ a \}
	& if $r = a$ and $a \in \Sigma$
	\\
\lang(r_1) \cup \lang(r_2)
	& if $r = r_1 + r_2$
	\\
\lang(r_1) \cdot \lang(r_2)
	& if $r = r_1 \cdot r_2$
	\\
(\lang(r'))^*
	& if $r = (r')^*$
\end{array}
\right.
\]

\noindent
A language $L \subseteq \Sigma^*$ is \emph{regular} if and only if there is a regular expression $r \in \rexp(\Sigma)$ such that $\lang(r) = L$.
\end{defi}

\subsection{Finite Automata}

Traditional accounts of finite automata typically introduce three variations of the notion:  deterministic (DFA), nondeterministic (NFA), and nondeterministic with $\varepsilon$-transitions (NFA-$\varepsilon$).  I will do the same, although I will do so in a somewhat different order than is typical.

\begin{defi}[Nondeterministic Finite Automaton (NFA)]\label{defi:nfa}
A \emph{nondeterministic finite automata} (NFA) is a tuple $(Q, \Sigma, \delta, q_I, F)$, where:
\begin{itemize}
\item
$Q$ is a finite non-empty set of \emph{states};
\item
$\Sigma$ is an \emph{alphabet};
\item
$\delta \subseteq Q \times \Sigma \times Q$ is the \emph{transition relation};
\item
$q_I \in Q$ is the \emph{initial state}; and
\item
$F \subseteq Q$ is the set of \emph{accepting}, or \emph{final}, states.
\end{itemize}
\end{defi}

This definition of NFA differs slightly from e.g.\/~\cite{hopcroftMU2006} in that $\delta$ is given as relation rather than  function in $Q \times \Sigma \rightarrow 2^Q$.  It also defines the form of a NFA but not the sense in which it is indeed a machine for processing words in a language.  The next definition does this by associating a language $\lang(M)$ with a given NFA $M = (Q, \Sigma, \delta, q_I, F)$.

\begin{defi}[Language of a NFA]\label{defi:nfa-language}
Let $M = (Q, \Sigma, \delta, q_I, F)$ be a NFA.
\begin{enumerate}
\item
Let $q \in Q$ be a state of $M$ and $w \in \Sigma^*$ be a word over $\Sigma$.  Then $M$ \emph{accepts} $w$ from $q$ if and only if one of the following holds.
\begin{itemize}
\item $w = \varepsilon$ and $q \in F$; or
\item $w = aw'$ some $a \in \Sigma$ and $w' \in \Sigma^*$, and there exists $(q,a,q') \in \delta$ such that $M$ accepts $w'$ from $q'$.
\end{itemize}
\item
The \emph{language}, $\lang(M)$, accepted by $M$ is defined as follows.
\[
\lang(M) = \{ w \in \Sigma^* \mid M \textnormal{ accepts } w \textnormal{ from } q_I \}
\]
\end{enumerate}
\end{defi}

Deterministic Finite Automata (DFAs) constitute a subclass of NFAs whose transition relation is deterministic, in a precisely defined sense.

\begin{defi}[Deterministic Finite Automaton (DFA)]
NFA $M = (Q, \Sigma, \delta, q_I, F)$ is a \emph{deterministic finite automaton} (DFA) if and only if $\delta$ satisfies the following:  for every $q \in Q$ and $a \in \Sigma$, there exists exactly one $q'$ such that $(q, a, q') \in \delta$.
\end{defi}

Since DFAs are NFAs the definition of $\lang$ in Definition~\ref{defi:nfa-language} is directly applicable to them as well.
NFAs with $\epsilon$-transitions are now defined as follows.

\begin{defi}[NFAs with $\varepsilon$-Transitions]
A \emph{nondeterministic automaton with $\varepsilon$-transitions} (NFA-$\varepsilon$) is a tuple $(Q, \Sigma, \delta, q_I, F)$, where:
\begin{itemize}
\item
$Q$ is a nonempty finite set of \emph{states};
\item
$\Sigma$ is an \emph{alphabet}, with $\varepsilon \not\in \Sigma$;
\item
$\delta \subseteq Q \times (\Sigma \cup \{\varepsilon\}) \times Q$ is the \emph{transition relation};
\item
$q_I \in Q$ is the \emph{initial state}; and
\item
 $F$ is the set of \emph{accepting}, or \emph{final}, states.
\end{itemize}
\end{defi}

An NFA-$\varepsilon$ is like a NFA except that some transitions can be labeled with the empty string $\varepsilon$ rather than a symbol from $\Sigma$.  The intution is that a transition of form $(q, \varepsilon, q')$ can occur without consuming any symbol as an input.  Formalizing this intuition, and defining $\lang(M)$ for NFA-$\varepsilon$, may be done as follows.

\begin{defi}[Language of a NFA-$\varepsilon$]
Let $M = (Q, \Sigma, \delta, q_I, F)$ be a NFA-$\varepsilon$.
\begin{enumerate}
\item
Let $q \in Q$ and $w \in \Sigma^*$.  Then $M$ \emph{accepts} $w$ \emph{from} $q$ if and only if one of the following holds.
\begin{itemize}
\item
$w = \varepsilon$ and $q' \in F$; or
\item
$w = aw'$ for some $a \in \Sigma$ and $w' \in \Sigma^*$ and there exists $q' \in Q$ such that $(q, a, q') \in \delta$ and $M$ accepts $w'$ from $q'$; or
\item
there exists $q' \in Q$ such that $(q, \varepsilon, q') \in \delta$ and $M$ accepts $w$ from $q'$.
\end{itemize}
\item
The \emph{language}, $\lang(M)$, accepted by $M$ is defined as follows.
\[
\lang(M) = \{ w \in \Sigma^* \mid M \textnormal{ accepts } w \textnormal{ from } q_I \}
\]
\end{enumerate}
\end{defi}

Defining the language of a NFA-$\varepsilon$ requires redefining the notion of a machine accepting a string from state $q$ as given in the definition of the language of a NFA.  This redefinition reflects the essential difference between $\varepsilon$-transitions and those labeled by alphabet symbols.

The three types of automata have differences in form, but equivalent expressive power.  It should first be noted that, just as every DFA is already a NFA, every NFA is also a NFA-$\varepsilon$, namely, a NFA-$\varepsilon$ with no $\varepsilon$-transitions.  Thus, every language accepted by some DFA is also accepted by some NFA, and every language accepted by some NFA is accepted by some NFA-$\varepsilon$.  The next theorem establishes the converses of these implications.

\begin{thm}[Equivalence of DFAs, NFAs and NFA-$\varepsilon$s]\label{thm:fa-equivalence}
\hfill
\begin{enumerate}
\item\label{case1}
Let $M$ be a NFA.  Then there is a DFA $D(M)$ such that $\lang(D(M)) = \lang(M)$.
\item\label{case2}
Let $M$ be a NFA-$\varepsilon$.  Then there is a NFA $N(M)$ such that $\lang(N(M)) = \lang(M)$.
\end{enumerate}
\end{thm}

\begin{proof}
The proof of Case~(\ref{case1}) involves the well-known subset construction, whereby each subset of states in $M$ is associated with a single state in $D(M)$.  The proof of Case~(\ref{case2}) typically relies on defining the $\varepsilon$ closure of a set of states, namely, the set of states reachable from the given set via a sequence of zero or more $\varepsilon$-transitions.  This notion is used to define the transition relation of $N(M)$ as well as its set of accepting states.
\end{proof}

\section{Kleene's Theorem}

Given the definitions in the previous section it is now possible to state Kleene's Theorem succinctly.

\begin{thm}[Kleene's Theorem]
Let $\Sigma$ be an alphabet.  Then $L \subseteq \Sigma^*$ is regular if and only if there is a DFA $M$ such that $\lang(M) = L$.
\end{thm}

The proof of this theorem is usually split into two pieces.  The first involves showing that for any regular expression $r$, there is a finite automaton $M$ (DFA, NFA or NFA-$\varepsilon$) such that $\lang(M) = \lang(r)$.  Theorem~\ref{thm:fa-equivalence} then ensures that the resulting finite automaton, if it is not already a DFA, can be converted into one in a language-preserving manner.  The second shows how to convert a DFA $M$ into a regular expression $r$ in such a way that $\lang(r) = \lang(M)$; there are several algorithms for this in the literature, including the classic dynamic-programming-based method of Kleene~\cite{kleene1956} and equation-solving methods that rely on Arden's Lemma~\cite{arden1961}.

From a practical standpoint, the conversion of regular expressions to finite automata is the more important, since regular expressions are textual and are used consequently as the basis for string search and processing.  For this reason, I believe that teaching this construction is especially keyin automata-theory classes, and this where my complaint with the approaches in traditional automata-theory texts originates.

To understand the basis for my dissatisfaction, let us review the construction presented in~\cite{hopcroftMU2006}, which explains how to convert regular expression $r$ into NFA-$\varepsilon$ $M_r$ in such a way that $\lang(r) = \lang(M_r)$.  The method is based on the construction due to Ken Thompson~\cite{thompson1968} and produces NFA-$\varepsilon$ $M_r$ with the following properties.
\begin{itemize}
\item 
The initial state $q_I$ has no incoming transitions:  that is, there exists no $(q, \alpha, q_I) \in \delta$.
\item
There is a single accepting state $q_F$, and $q_F$ has no outgoing transitions:  that is, $F = \{ q_F \}$, and there exists no $(q_F, \alpha, q') \in \delta$.
\end{itemize}
The approach proceeds inductively on the structure of $r$.  For example, if $r = (r')^*$, then assume that $M_{r'} = (Q, \Sigma, \delta, q_I, \{ q_F \})$ meeting the above constraints has been constructed.  Then $M_r$ is built as follows.  First, let $q_{I}' \not\in Q$ and $q_{F}' \not\in Q$ be new states.  Then $M_r = (Q \cup \{q_I', q_F'\}, \Sigma, \delta', \{q_F' \})$, where
\[
\delta' = \delta \cup \{ (q_I', \varepsilon, q_I), (q_I', \varepsilon, q_F'), (q_F, \varepsilon, q_I), (q_F, \varepsilon, q_F') \}.
\]
It can be shown that $M_r$ satisfies the requisite properties and that $\lang(M_r) = (\lang(r'))^*$.

Mathematically, the construction of $M_r$ is wholly satisfactory:  it has the required properties and can be defined relatively easily, albeit at the cost of introducing new states and transitions.  The proof of correctness is perhaps somewhat complicated, owing to the definition of $\lang(M)$ and the subtlety of $\varepsilon$-transitions, but it does acquaint students with definitions via structural induction on regular expressions.

My concern with the construction, however, is several-fold.  On the one hand, it does require the introduction of the notion of NFA-$\varepsilon$, which is indeed more complex that that of NFA.  In particular, the definition of acceptance requires allowing transitions that consume no symbol in the input word.  On the other hand, the accretion of the introduction of new states at each state in the construction makes it difficult to test students on their understanding of the construction in an exam setting.  Specifically, even for relatively small regular expressions the literal application of the construction yields automata with too many states and transitions to be doable during the typical one-hour midterm exam for which US students would be tested on the material.  Finally, the construction bears no resemblance to algorithms used in practice for construction finite automata from regular expressions.  In particular routines such as the Berry-Sethi procedure~\cite{berry1986regular} construct DFAs directly from regular expressions, completely avoiding the need for NFA-$\varepsilon$s, or indeed NFAs, altogether.

The Berry-Sethi procedure is subtle and elegant, and relies on concepts, such as Brzozowski derivatives~\cite{brzozowski1964derivatives}, that I would view as too specialized for an undergraduate course on automata theory.  Consequently, I would not be in favor of covering them in an undergraduate classroom setting.  Instead, in the next section I give a technique, based on operational semantics in process algebra, for construction NFAs from regular expressions.  The resulting NFAs are small enough for students to construct during exams, and the construction has other properties, including the capacity for introducing other operations that preserve regularity, that are pedagogically useful.

\section{NFAs via Structural Operational Semantics}


This section describes an approach based on \emph{Structural Operational Semantics} (SOS)~\cite{plotkin1981,plotkin2004} for constructing NFAs from regular expressions.  Specifically, I will define a (small-step) operational semantics for regular expressions on the basis of the structure of regular expressions, and use the semantics to construct the requisite NFAs.  The construction requires no $\varepsilon$-transitions and yields automata with at most one more state state than the size of the regular expression from which they are derived.

Following the conventions in the other parts of this paper I give the SOS rules using notation typically found in automata-theory texts.  In particular, the SOS specification is given in natural language, as a collection of if-then statements, and not via inference rules.  I use this approach in the classroom to avoid having to introduce notations for inference rules.  In the appendix I give the more traditional SOS presentation.

\subsection{An Operational Semantics for Regular Expressions}

In what follows fix alphabet $\Sigma$.
The basis for the operational semantics of regular expressions consists of a relation, $\derives{} \subseteq \rexp(\Sigma) \times \Sigma \times \rexp(\Sigma)$, and a predicate $\surd \subseteq \rexp(\Sigma)$.  In what follows I will write $r \derives{a} r'$ and $r \surd$ in lieu of $(r, a, r') \in \,\derives{}$ and $r \in \surd$.  The intuitions are as follows.
\begin{enumerate}
\item
$r \surd$ is intended to hold if and only if $\varepsilon \in \lang(r)$.  This is used in defining accepting states.
\item
$r \derives{a} r'$ is intended to reflect the following about $\lang(r)$:  one way to build a word in $\lang(r)$ is to start with $a \in \Sigma$ and then finish it with a word from $\lang(r')$.

\end{enumerate}

Using these relations, I then show how to build a NFA from $r$ whose states are regular expressions, whose transitions are given by $\derives{}$, and whose final states are defined using $\surd$.

\paragraph{\textbf{\textit{Defining \texorpdfstring{$\surd$}{Acceptance predicate for regular expressions} and \texorpdfstring{$\derives{}$}{Transition relation for regular expressions}}}}

We now define $\surd$.
\begin{defi}[Definition of $\surd$]\label{def:surd}
Predicate $r \surd$ is defined inductively on the structure of $r \in \rexp(\Sigma)$ as follows.
\begin{itemize}
\item If $r = \varepsilon$ then $r \surd$.
\item If $r = (r')^*$ for some $r' \in \rexp(\Sigma)$ then $r \surd$.
\item If $r = r_1 + r_2$ for some $r_1, r_2 \in \rexp(\Sigma)$, and $r_1 \surd$, then $r \surd$.
\item If $r = r_1 + r_2$ for some $r_1, r_2 \in \rexp(\Sigma)$, and $r_2 \surd$, then $r \surd$.
\item If $r = r_1 \cdot r_2$ for some $r_1, r_2 \in \rexp(\Sigma)$, and $r_1 \surd$ and $r_2 \surd$, then $r \surd$.
\end{itemize}
\end{defi}

From the definition, one can see it is not the case that $\emptyset\surd$ or $a \surd$, for any $a \in \Sigma$, while both $\varepsilon \surd$ and $r^* \surd$ always.  This accords with the definition of $\lang(r)$; $\varepsilon \not\in \lang(\emptyset) = \emptyset $, and $\varepsilon \not\in \lang(a) = \{ a \}$, while $\varepsilon \in \lang(\varepsilon) = \{ \varepsilon \}$ and $\varepsilon \in L^*$ for any language $L \subseteq \Sigma^*$, and in particular for $L = \lang(r)$ for regular expression $r$.  The other cases in the definition reflect the fact that $\varepsilon \in \lang(r_1 + r_2)$ can only hold if $\varepsilon \in \lang(r_1)$ or $\varepsilon \in \lang(r_2)$, since $+$ is interpreted as set union, and that $\varepsilon \in \lang(r_1 \cdot r_2)$ can only be true if $\varepsilon \in \lang(r_1)$ and $\varepsilon \in \lang(r_2)$, since regular-expression operator $\cdot$ is interpreted as language concatenation.  We have the following examples.
\[
\begin{array}{lp{3in}}
(\varepsilon \cdot a^*)\surd & since $\varepsilon \surd$ and $a^* \surd$.\\
\neg(a + b)\surd  & since neither $a \surd$ nor $b \surd$.\\
(01 + (1+01)^*)\surd  & since $(1+01)^* \surd$.\\
\neg(01(1+01)^*)\surd  & since $\neg(01)\surd$.\\
\end{array}
\]

\noindent
We also use structural induction to define $\derives{}$.
\begin{defi}[Definition of $\derives{}$]\label{def:derives}
Relation $r \derives{a} r'$, where $r, r' \in \rexp(\Sigma)$ and $a \in \Sigma$, is defined inductively on $r$.
\begin{itemize}
\item If $r = a$ and $a \in \Sigma$ then $r \derives{a} \varepsilon$.
\item If $r = r_1 + r_2$ and $r_1 \derives{a} r_1'$ then $r \derives{a} r_1'$.
\item If $r = r_1 + r_2$ and $r_2 \derives{a} r_2'$ then $r \derives{a} r_2'$.
\item If $r = r_1 \cdot r_2$ and $r_1 \derives{a} r_1'$ then $r \derives{a} r_1' \cdot r_2$.
\item If $r = r_1 \cdot r_2$, $r_1 \surd$ and $r_2 \derives{a} r_2'$ then $r \derives{a} r_2'$.
\item If $r = (r')^*$ and $r' \derives{a} r''$ then $r \derives{a} r'' \cdot (r')^*$.
\end{itemize}
\end{defi}
The definition of this relation is somewhat complex, but the idea that it is
trying to capture is relatively simple: $r \derives{a} r'$ if one can build
words in $\lang(r)$ by taking the $a$ labeling $\derives{}$ and appending a
word from $\lang(r')$.  So we have the rule $a \derives{a} \varepsilon$ for $a \in
\Sigma$, while the rules for $+$ follow from the fact that $\lang(r_1 + r_2) =
\lang(r_1) \cup \lang(r_2)$.  The cases for $r_1 \cdot r_2$ in essence state that $aw \in
\lang(r_1 \cdot r_2)$ can hold either if there is a way of splitting $w$ into $w_1$ and $w_2$
such that $aw_1$ is in the language of $r_1$ and $w_2$ is in the language of $r_2$,
or if $\varepsilon$ is in the language of $r_1$ and $aw$ is in the language of $r_2$.
Finally, the rule for $(r')^*$ essentially permits ``looping''.  As examples, we have the following.
\[
\begin{array}{lp{8cm}}
a+b \derives{a} \varepsilon & by the rules for $a$ and $+$.\\
(abb + a)^* \derives{a} \varepsilon bb(abb+a)^*
    & by the rules for $a$, $\cdot$, $+$, and $^*$.
\end{array}
\]

\noindent
In this latter example, note that applying the definition literally requires the inclusion of the $\varepsilon$ in $\varepsilon bb(abb+a)^*$.  This is because the case for
$a$ says that $a \derives{a} \varepsilon$, meaning that $abb \derives{a} \varepsilon
bb$, etc.  However, when there are leading instances of $\varepsilon$ like this, I will sometimes leave them
out, and write $abb \derives{a} bb$ rather than $abb \derives{a}
\varepsilon bb$.\footnote{This convention can be formalized by introducing a special case in the definition of $\derives{}$ for $a \cdot r_2$ and distinguishing the current two cases for $r_1 \cdot r_2$ to apply only when $r_1 \not\in \Sigma.$}

The following lemmas about $\surd$ and $\derives{}$ formally establish the
intuitive properties that they should have.

\begin{lem}\label{lem:surd-lemma}
Let $r \in \rexp(\Sigma)$ be a regular expression.  Then $r \surd$ if and only if $\varepsilon \in
\lang(r)$.
\end{lem}

\begin{proof}
The proof proceeds by structural induction on $r$.  Most cases are left to the reader; we only consider the $r = r_1 \cdot r_2$ case here.  The induction hypothesis states that $r_1 \surd$ if and only if $\varepsilon \in \lang(r_1)$ and $r_2 \surd$ if and only if $\varepsilon \in \lang(r_2)$.  One reasons as follows.
\[
\begin{array}{r@{\textnormal{ iff }}lp{6cm}}
r \surd
	& r_1 \surd \textnormal{ and } r_2 \surd
	& Definition of $\surd$
\\
	& \varepsilon \in \lang(r_1) \textnormal{ and } \varepsilon \in \lang(r_2)
	& Induction hypothesis
\\
	& \varepsilon \in (\lang(r_1)) \cdot (\lang(r_2))
	& Property of concatenation
\\
	& \varepsilon \in \lang(r_1 \cdot r_2)
	& Definition of $\lang(r_1 \cdot r_2)$
\\
	& \varepsilon \in \lang(r)
	& $r = r_1 \cdot r_2$
\end{array}
\]\qedhere
\end{proof}

\begin{lem}\label{lem:derives-lemma}
Let $r \in \rexp(\Sigma)$, $a
\in \Sigma$, and $w \in \Sigma^*$.  Then $aw \in \lang(r)$ if and only
if there is an $r' \in \rexp(\Sigma)$ such that $r \derives{a} r'$ and $w \in
\lang(r')$.
\end{lem}

\begin{proof}
The proof proceeds by structural induction on $r$.  We only consider the case $r = (r')^*$ in detail; the others are left to the reader.  The induction hypothesis asserts that for all $a$ and $w'$, $aw' \in \lang(r')$ if and only if there is an $r''$ such that $r' \derives{a} r''$ and $w' \in \lang(r'')$.  We reason as follows.
\[
\begin{array}{r@{\textnormal{ iff }}lp{4.8cm}}
aw \in \lang(r)
	& aw \in \lang((r')^*)
	& $r = (r')^*$
\\
	& aw \in (\lang(r'))^*
	& Definition of $\lang((r')^*)$
\\
	& aw = w_1 \cdot w_2 \textnormal{ some } w_1 \in \lang(r'), w_2 \in (\lang(r'))^*
	& Definition of Kleene closure
\\
	& w_1 = a \cdot w_1' \textnormal{ some } w_1'
	& Property of Kleene closure
\\
	& r' \derives{a} r'' \textnormal{ some } r'' \textnormal{ with } w_1' \in \lang(r'')
	& Induction hypothesis
\\
	& r \derives{a} r'' \cdot (r')^*
	& Definition of $\derives{}$
\\
	& w_1' \cdot w_2 \in \lang(r'') \cdot \lang((r')^*)
	& Definition of concatenation
\\
	& w_1' \cdot w_2 \in \lang(r'' \cdot (r')^*)
	& Definition of $\lang(r'' \cdot (r')^*)$
\\
	& r \derives{a} r'' \cdot (r')^* \textnormal{ and } w \in \lang(r'' \cdot (r')^*)
	& $w = w_1' \cdot w_2$
\end{array}
\]
\end{proof}

Appendix~\ref{app} contains definitions of $\surd$ and $\derives{}$ in the more usual inference-rule style used in SOS specifications.

\subsection{Building Automata using \texorpdfstring{$\surd$}{Acceptance} and \texorpdfstring{$\derives{}$}{Transitions}}

\noindent
That $\surd$ and $\derives{}$ may be used to build NFAs derives from how
they may be used to determine
whether a string is in the language of a regular expression.  Consider the
following sequence of transitions starting from the regular expression
$(abb + a)^*$.
\[
(abb + a)^* \derives{a} bb(abb+a)^* \derives{b} b(abb+a)^* \derives{b}
(abb+a)^* \derives{a} (abb+a)^*
\]

Using Lemma~\ref{lem:derives-lemma} four times, we can conclude that if $w \in
\lang((abb+a)^*)$, then $abba \cdot w \in \lang((abb+a)^*)$ also.  In addition,
since $(abb+a)^* \surd$, it follows from Lemma~\ref{lem:surd-lemma} that $\varepsilon \in
\lang((abb+a)^*)$.  Since $abba \cdot \varepsilon = abba$, it follows that
$abba \in \lang((abb+a)^*)$.

More generally, if there is a sequence of
transitions $r_0 \derives{a_1} r_1 \cdots \derives{a_n} r_n$ and $r_n \surd$,
then it follows that $a_1 \ldots a_n \in \lang(r_0)$, and vice versa.
This observation suggests the following strategy for building a NFA
from a regular expression $r$.
\begin{enumerate}
\item\label{states} Let the states be all possible regular expressions that can
    be reached by
    some sequence of transitions from $r$.
\item Take $r$ to be the start state.
\item Let the transitions be given by $\derives{}$.
\item Let the accepting states be those regular expressions $r'$ reachable from $r$ for which
    $r'\surd$ holds.
\end{enumerate}
Of course, this construction is only valid if the set of all possible regular
expressions mentioned in Step~(\ref{states}) is finite, since NFAs are required
to have a finite number of states.  In fact, a stronger result can be proved.  First, recall the definition of the size, $|r|$, of regular expression $r$.

\begin{defi}[Size of a regular expression]
The size, $|r|$, of $r \in \rexp(\Sigma)$ is defined inductively as follows.
\[
|r| =
\left\{
\begin{array}{lp{8cm}}
1
	& if $r = \varepsilon, r = \emptyset,$ or $r = a$ for some $a \in \Sigma$
\\
|r'| + 1
	& if $r = (r')^*$
\\
|r_1| + |r_2| + 1
	& if $r = r_1 + r_2$ or $r = r_1 \cdot r_2$
\end{array}
\right.
\]
\end{defi}

\noindent
Intuitively, $|r|$ counts the number of regular-expression operators in $r$.  The \emph{reachability set} of regular expression $r$ can now be defined in the usual manner.

\begin{defi}
Let $r \in \rexp(\Sigma)$ be a regular expression.  Then the set $RS(r) \subseteq
\rexp(\Sigma)$ of regular expressions \emph{reachable from} $r$  is defined recursively as follows.
\begin{itemize}
\item $r \in RS(r)$.
\item If $r_1 \in RS(r)$ and $r_1 \derives{a} r_2$ for some $a \in \Sigma$,
    then $r_2 \in RS(r)$.
\end{itemize}
\end{defi}

\noindent
As an example, note that $|(abb + a)^*| = 8$ and that
\[
RS((abb +
a)^*) = \{ (abb + a)^*, \varepsilon bb(abb+a)^*, \varepsilon b(abb+a)^*, \varepsilon(abb + a)^* \},
\]
(In this case I have not applied my heuristic of suppressing leading $\varepsilon$ expressions.)
The following can now be provd.
\begin{thm}\label{thm:num-states}
Let $r \in \rexp(\Sigma)$ be a regular expression.  Then $|RS(r)| \leq |r|+1$.
\end{thm}

\begin{proof}
The proof proceeds by structural induction on $r$.  There are six cases to consider.
\begin{description}
\item[$r = \emptyset$]  In this case $RS(r) = \{ \emptyset \}$, and $|RS(r)| = 1 = |r| < |r|+1$.
\item[$r = \varepsilon$]  In this case $RS(r) = \{ \varepsilon \}$, and $|RS(r)| = 1 = |r| < |r|+1$.
\item[$r = a$ for some $a \in \Sigma$]
	In this case $RS(r) = \{a, \varepsilon \}$, and $|RS(r)| = 2 = |r|+1$.
\item[$r = r_1 + r_2$]
	In this case, $RS(r) \subseteq RS(r_1) \cup RS(r_2)$, and the induction hypothesis guarantees that $|RS(r_1)| \leq |r_1| + 1$ and $RS(r_2) \leq |r_2| + 1$.  It then follows that
	\[
	|RS(r)| \leq |RS(r_1)| + |RS(r_2)| \leq |r_1| + |r_2| + 2 = |r| + 1.
	\]
\item[$r = r_1 \cdot r_2$]
	In this case it can be shown that $RS(r) \subseteq \setof{r_1' \cdot r_2}{r_1' \in RS(r_1)} \cup RS(r_2)$.  Since 
	$|\setof{r_1' \cdot r_2}{r_1' \in RS(r_1)}| = |RS(r_1)|$, similar reasoning as in the $+$ case applies.
\item[$r = (r')^*$]
	In this case we have that $RS(r) \subseteq \{ r \} \cup \setof{r'';r}{r'' \in RS(r')}$.  Thus
	\[
	|RS(r)| \leq |RS(r')| + 1 \leq |r'|+2 = |r| + 1.
	\]
\end{description}
\end{proof}

This result shows not only that the sketched NFA construction given above yields a finite number of states for given $r$, it in fact establishes that this set of state is no larger than $|r| + 1$.  This highlights one of the main reasons I opted to introduce this construction in my classes:  small regular expressions yield NFAs that are almost as small, and can be constructed manually in an exam setting.

We can now formally define the construction of NFA $M_r$ from regular expression $r$ as follows.
\begin{defi}\label{def:Mr}
Let $r \in \rexp(\Sigma)$ be a regular expression.  Then $M_r = (Q,
\Sigma, q_I, \delta, A)$ is the NFA defined as follows.
\begin{itemize}
\item $Q = RS(r)$.
\item $q_I = r$.
\item $\delta = \setof{(r_1, a, r_2)}{r_1 \derives{a} r_2}$.
\item $F = \setof{r' \in Q}{r' \surd}$.
\end{itemize}
\end{defi}

The next theorem establishes that $r$ and $M_r$
define the same languages.
\begin{thm}
Let $r \in \rexp(\Sigma)$ be a regular expression.  The $\lang(r) = \lang(M_r)$.
\end{thm}
\begin{proof}
Relies on the fact that Lemmas~\ref{lem:surd-lemma} and~\ref{lem:derives-lemma}
guarantee that $w = a_1\ldots a_n \in \lang(r)$ if and only if there is a
regular expression $r'$ such that $r \derives{a_1} \cdots \derives{a_n} r'$ and
$r'\surd$.
\end{proof}

\subsection{Computing \texorpdfstring{$M_r$}{NFAs from Regular Expressions}}

This section gives a routine for computing $M_r$.  It intertwines the computation of the reachability set from regular expression $r$ with the updating of the transition relation and set of accepting states.  It relies on the computation of the so-called \emph{outgoing transitions} of $r$; these are defined as follows.

\begin{defi}
Let $r \in \rexp(\Sigma)$ be a regular expression.  Then the set of
\emph{outgoing transitions} from $r$ is defined as the set
$\setof{(a,r')}{r \derives{a} r'}$.
\end{defi}
The outgoing transitions from $r$ consists of pairs
$(a,r')$ that, when combined with $r$, constitute a valid
transition $r \derives{a} r'$.
Figure~\ref{fig:out-def} defines a recursive function, $\out$, for computing the outgoing transitions of $r$.
The routine uses the structure of $r$ and the definition of $\derives{}$ to guide its computation.  For regular
expressions of the form $\emptyset, \varepsilon$ and $a \in \Sigma$, the
definition of $\derives{}$ in Definition~\ref{def:derives} immediately gives all the transitions.  For
regular expressions built using $+, \cdot$ and $^*$, one must first
recursively compute the outgoing transitions of the subexpressions of $r$ and
then combine the results appropriately, based on the cases given in the
Definition~\ref{def:derives}.

\begin{figure}
$$
\out(r) =
\left\{
\begin{array}{lp{10em}}
\emptyset
    & if $r = \emptyset$ or $r = \varepsilon$\\
\{(a, \varepsilon)\}
    & if $r = a \in \Sigma$\\
\out(r_1) \cup  \out(r_2)
    & if $r = r_1 + r_2$\\
\setof{(a, r_1' \cdot r_2)}{(a,r_1') \in \out(r_1)}
    &\\
\;\;\;\;\;\;\cup\; 
\setof{(a, r_2')}{(a, r_2') \in \out(r_2) \land r_1 \surd}
    & if $r = r_1 \cdot r_2$\\
\setof{(a, r'' \cdot (r')^*)}{(a, r_1') \in \out(r_1)}
    & if $r = (r')^*$
\end{array}
\right.
$$
\caption{Calculating the outgoing transitions of regular expressions.}
\label{fig:out-def}
\end{figure}

The next lemma states that $\out(r)$ correctly computes the outgoing
transitions of $r$.
\begin{lem}
Let $r \in \rexp(\Sigma)$ be a regular expression, and let $\out(r)$ be as defined 
in Figure~\ref{fig:out-def}.  Then $\out(r) = \setof{(a,r')}{r
\derives{a} r'}$.
\end{lem}
\begin{proof}
By structural induction on $r$.  The details are left to the reader.
\end{proof}

\noindent
Algorithm~\ref{alg:nfa} contains pseudo-code for computing $M_r$.  It maintains four sets.
\begin{itemize}
\item $Q$, a set that will eventually contain the states of $M_r$.
\item $F$, a set that will eventually contain the accepting states of $M_r$.
\item $\delta$, a set that will eventually contain the transition relation of $M_r$.
\item $W$, the \emph{work set}, a subset of $Q$ containing states that have not yet had their
    outgoing transitions computed or acceptance status determined.
\end{itemize}
The procedure begins by adding $r$, its input parameter, to both $Q$ and $W$.  It then repeatedly removes a state from $W$, determines if it should be added to $F$, computes its outgoing transitions and updates $\delta$ appropriately, and finally adds the target states in the outgoing transition set to both $Q$ and $W$ if they are not yet in $Q$ (meaning they have not yet been encountered in the construction of $M_r$).  The algorithm terminates when $W$ is empty.

\begin{algorithm}[t]
  \SetKwInOut{Input}{Input}
  \SetKwInOut{Output}{Output}
  \DontPrintSemicolon 

  \textbf{Algorithm} \textit{NFA}$(r)$\;
  \Input{Regular rexpression $r \in \rexp(\Sigma)$}
  \Output{NFA $M_r = (Q, \Sigma, q_I, \delta, F)$}
  
    $Q := \{ r \}$ \tcp*{State set}
    $q_I := r$ \tcp*{Start state}
    $W := \{ r \}$ \tcp*{Working set}
    $\delta := \emptyset$ \tcp*{Transition relation}
    $F := \emptyset$ \tcp*{Accepting states}

    \While{$W \neq \emptyset$}{
      choose $r' \in W$\;
      $W := W - \{r'\}$\;
      \uIf{$r' \surd$ }{
        $F := F \cup \{r'\}$  \tcp*[f]{$r'$ is an accepting state}
      }
      $T = \out(r')$  \tcp*{Outgoing transitions of $r'$}
      $\delta := \delta \cup \setof{r', a, r'')}{(a, r'') \in T}$ \tcp*{Update transition relation}
      
      \ForEach{$(a, r'') \in T$}{
        \uIf{$r'' \not\in Q$}{ 
          $Q := Q \cup \{ r'' \}$  \tcp*{$r''$ is a new expression}
          $W := W \cup \{ r'' \}$
        }
      }
    }

    \Return $M_r = (Q, \Sigma, \delta, q_I, F)$
  \caption{Algorithm for computing NFA $M_r$ from regular expression $r$}\label{alg:nfa}

\end{algorithm}

Figure~\ref{fig:nfa-example}
gives the NFA resulting from applying the procedure to $(abb + a)^*$.  Figure~\ref{fig:nfae-example}, by way of contrast, shows  the result of applying the routine in~\cite{hopcroftMU2006} to produce a NFA-$\varepsilon$ from the same regular expression.

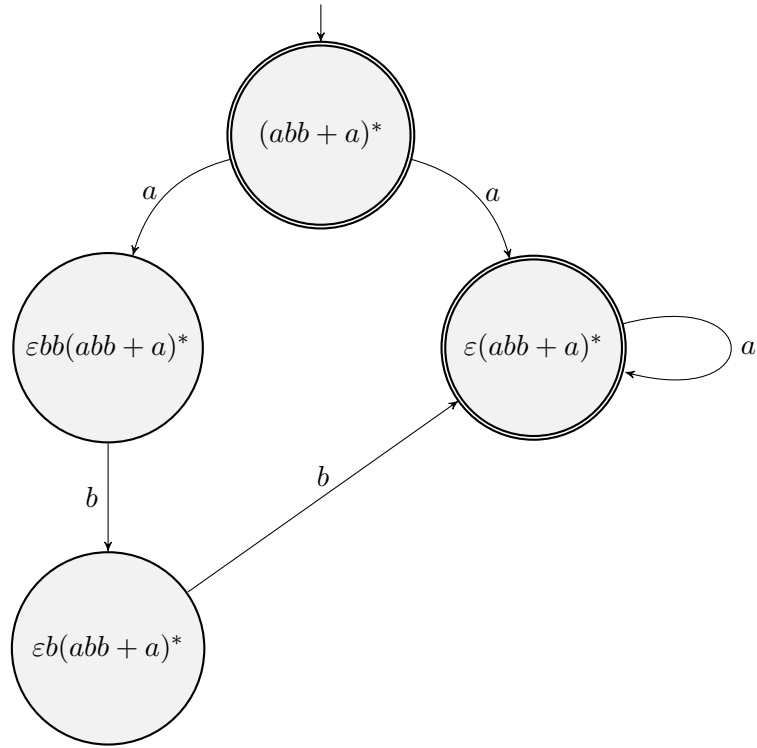
\begin{figure}
\tikzset{
  ->, 
  >=stealth', 
  node distance=4cm, 
  every state/.style={thick, fill=gray!10} 
}
\centering
\begin{tikzpicture}
	\node[state, initial, initial where=above, accepting] (q0) {$\;\;(abb + a)^*\;\;$};
	\node[state, below left of=q0] (q1) {$\varepsilon bb (abb + a)^*$};
	\node[state, below of=q1] (q2) {$\;\varepsilon b (abb + a)^*\;$};
	\node[state, accepting, below right of=q0] (q3) {$\;\varepsilon (abb + a)^*\;$};

	\draw	(q0) edge[bend right, left] node{$a$} (q1)
			(q0) edge[bend left, right] node{$a$} (q3)
			(q1) edge[left] node{$b$} (q2)
			(q2) edge[above] node{$b$} (q3)
			(q3) edge[loop right, right] node{$a$} (q3);
	
\end{tikzpicture}
\caption{NFA$(r)$ for $r = (abb + a)^*$.}
\label{fig:nfa-example}
\end{figure}

\begin{figure}
\centering
\begin{tikzpicture}
\tikzset{
  ->, 
  >=stealth', 
  node distance=2cm, 
  every state/.style={thick, fill=gray!10} 
}
	\node[state,initial] (10) {};
	\node[state, below of=10] (8) {};
	\node[state, below left of=8] (0) {};
	\node[state, below of=0] (1) {};
	\node[state, below of=1] (2) {};
	\node[state, below of=2] (3) {};
	\node[state, below of=3] (4) {};
	\node[state, below of=4] (5) {};
	\node[state, below right of=8] (6) {};
	\node[state, below of=6] (7) {};
	\node[state, below right of=5] (9) {};
	\node[state, accepting, right of=9] (11) {};

	\draw
		(8) edge[left] node{$\varepsilon$} (0)
		(8) edge[right] node{$\varepsilon$} (6)
		(0) edge[left] node{$a$} (1)
		(1) edge[left] node{$\varepsilon$} (2)
		(2) edge[left] node{$b$} (3)
		(3) edge[left] node{$\varepsilon$} (4)
		(4) edge[left] node{$b$} (5)
		(6) edge[right] node{$a$} (7)
		(5) edge[left] node{$\varepsilon$} (9)
		(7) edge[right] node{$\varepsilon$} (9)
		(9) edge[left] node{$\varepsilon$} (8)
		(9) edge[above] node{$\varepsilon$} (11)
		(10) edge[right] node{$\varepsilon$} (8)
		(10) edge[bend left, right] node{$\varepsilon$} (11)
	;
\end{tikzpicture}
\caption{NFA-$\varepsilon$ for $(abb + a)^*$.}
\label{fig:nfae-example}
\end{figure}
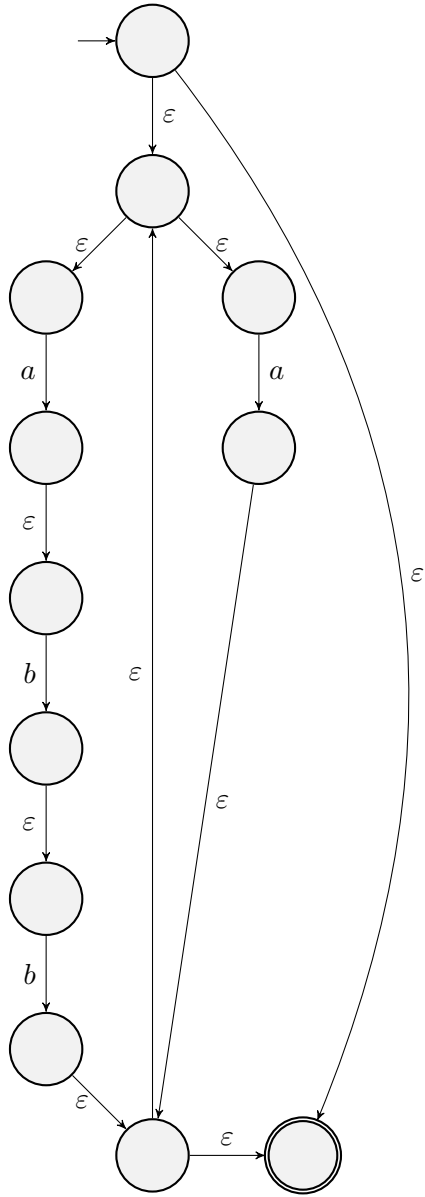

\section{Discussion}

The title of this note is ``Better Automata through Process Algebra,'' and I want to revisit it in order to explain in what respects I regard the method presented in here as producing ``better automata.''  Earlier I identified the following motivations that prompted me to incorporate this approach in my classroom instruction.

\begin{itemize}
\item
I wanted to produce NFAs rather than NFA-$\varepsilon$s.  In large part this was due to my desire not cover the notion of NFA-$\varepsilon$.  The only place this material is used in typical automata-theory textbooks is as a vehicle for converting regular expressions into finite automata.  By giving a construction that avoids the use of $\varepsilon$-transitions, I could avoid covering NFA-$\varepsilon$s and devote the newly freed lecture time to other topics.  Of course, this is only possible if the NFA-based construction does not require more time to describe than the introduction of NFA-$\varepsilon$ and the NFA-$\varepsilon$ construction.
\item
I wanted the construction to be one that students could apply during an exam to generate finite automata from regular expressions.  The classical construction found in~\cite{hopcroftMU2006} and other books fails this test, in my opinion; while the inductive definitions are mathematically pleasing, they yield automata with too many states for students to be expected to apply them in a time-constrained setting.
\item
Related to the preceding point, I wanted a technique that students could imagine being implemented and used in the numerous applications to which regular expressions are applied.  In such a setting, fewer states is better than more states, all things considered.
\end{itemize}

This note has attempted to argue these points by giving a construction in Definition~\ref{def:Mr} for constructing  NFAs directly from regular expressions.  Theorem~\ref{thm:num-states} estabishes that the number of states in these NFAs is at most one larger than the size of the regular expression from which the NFAs are generated; this provides guidance in preparing exam questions, as the size of the NFAs students can be asked to generate are tightly bounded by the size of the regular expression given in the exam.  Finally, Algorithm~\ref{alg:nfa} gives a ``close-to-code'' account of the construction that hints at its implementability.  Indeed, several years ago a couple of students that I presented this material to independently implemented the algorithm.

Beyond the points mentioned above, I think this approach has two other points in its favor.  The first is that is provides a basis for defining other operators over regular expressions and proving that the class of regular languages is closed with result to these operations.  The ingredients for introducing such a new operator and proving closure of regular languages with respect to it can be summarized as follows.
\begin{enumerate}
\item
Extend the definition of $\lang(r)$ given in Definition~\ref{def:rexp-language} to give a language-theoretic semantics for the operator.
\item
Extend the definitions of $\surd$ and $\derives{}$ in Definitions~\ref{def:surd} and~\ref{def:derives} to give a small-step operations semantics for the operator.
\item
Extend the proofs of Lemmas~\ref{lem:surd-lemma} and~\ref{lem:derives-lemma} to establish connections between the language semantics and the operational semantics.
\item
Prove that expressions extended with the new operator yield finite sets of reachable expressions.
\end{enumerate}
All of these steps involve adding new cases to the existing definitions and lemmas, and altering Theorem~\ref{thm:num-states} in the case of the last point.  Once these are done, Algorithm~\ref{alg:nfa}, with the definition of $\out$ given in Figure~\ref{fig:out-def} suitably modified to cover the new operator, can be used as is as a basis for constructing NFAs from these extended classes of regular languages.

I have used parts of this approach in the classroom to ask students to prove that synchronous product and interleaving operators can be shown to preserve language regularity.  Other operators, such as ones from process algebra, are also candidates for these kinds of questions.

The second feature of the approach in this paper that I believe recommends it is that the NFA construction is ``on-the-fly''; the construction of a automaton from a regular expression does not require the \emph{a priori} construction of automata from subexpressions, meaning that the actual production of the automaton can be intertwined with other operations, such as the checking of whether a word belongs to the regular expression's language.  One does not need to wait the construction of the full automaton, in other words, before putting it to use.

Criticisms that I have heard of this approach center around two issues.  The first is that the construction of NFA $M_r$ from regular expression $r$ does not use structural induction on $r$, unlike the classical constructions in e.g.~\cite{hopcroftMU2006}.  I do not have much patience with the complaint, as the concepts that $M_r$ is built on, namely $\surd$ and $\derives{}$, are defined inductively, and the results proven about them require substantial use of induction.  The other complaint is that the notion of $r \derives{a} r'$ is ``hard to understand.''  It is indeed the case that equipping regular expressions with an operational semantics is far removed from the language-theoretic semantics typically given to these expressions.  That said, I would argue that the small-step operational semantics considered here in fact exposes the essence of the relationship between regular expressions and finite automata:  this semantics enables regular expressions to be executed, and in a way that can be captured via automata.

I close this section with a brief discussion of the Berry-Sethi algorithm~\cite{berry1986regular}, which is used in practice and produces deterministic finite automata.  This feature enables their technique to accommodate complementation, an operation with respect to which regular languages are closed but which fits uneasily with NFAs.  From a pedagogical perspective, however, the algorithm suffers somewhat as number of states in a DFA can be exponentially larger than that size of the regular expression from which it is derived.  A similar criticism can be made of other techniques that rely on Brzozowsky derivatives~\cite{brzozowski1964derivatives}, which also produce DFAs.  There are interesting connections between our operational semantics and these derivatives, but we exploit nondeterminacy to keep the sizes of the resulting finite automata small.

\section{Conclusions and Directions for Future Work}

In this note I have presented an alternative approach for converting regular expressions into finite automata.  The method relies on defining an operational semantics for regular expressions, and as such draws inspiration from the work on process algebra undertaken by pioneers in that field, including Jos Baeten.  In contrast with classical techniques, the construction here does not require transitions labeled by the empty word $\varepsilon$, and it yields automata whose state sets are proportional in size to the regular expressions they come from.  The procedure can also be implemented in an on-the-fly manner, meaning that the production of the automaton can be intertwined with other analysis procedures as well.

Other algorithms studied in process algebra also have pedagogical promise, in my opinion.  One method, the Kanellakis-Smolka algorithm for computing bisimulation equivalence~\cite{kanellakis1990ccs}, is a case in point.  Partition-refinement algorithms for computing langauge equivalence of deterministic automata have been in existence for decades, but the details underpinning them are subtle and difficult to present in an undergraduate automata-theory class, where instructional time is at a premium.  While not as efficient asymptotically as the best procedures, the simplicity of the K-S technique recommends it, in my opinion, both for equivalence checking and state-machine minimization.  Simulation-checking algorithms~\cite{henzinger1995computing} can also be used as a basis for checking language containment among finite automata; these are interesting because they do not require determinization of both automata being compared, in general.

\bibliographystyle{alpha}
\bibliography{local}

\appendix
\section{SOS Rules for \texorpdfstring{$\surd$}{Acceptance} and \texorpdfstring{$\derives{}$}{Transitions}}\label{app}

Here are the inference rules used to define $\surd$.  They are given in the form
\[
\postrule{ \textit{premises} }{ \textit{conclusion} }{}
\]
with $-$ denoting an empty list of premises.
\[
\postrule{ - }{ \varepsilon \surd}{}
\;\;\;\;
\postrule{ - }{ r^* \surd }{}
\;\;\;\;
\postrule{ r_1 \surd }{ (r_1 + r_2) \surd }{}
\;\;\;\;
\postrule{ r_2 \surd }{ (r_1 + r_2) \surd }{}
\;\;\;\;
\postrule{ r_1 \surd \;\;\;\; r_2 \surd }{ (r_1 \cdot r_2) \surd }{}
\]

\noindent
Next are the rules for $\derives{}$.
\[
\postrule{-}{a \derives{a} \varepsilon}{}
\;\;\;\;
\postrule{r_1 \derives{a} r_1'}{r_1 + r_2 \derives{a} r_1'}{}
\;\;\;\;
\postrule{r_2 \derives{a} r_2'}{r_1 + r_2 \derives{a} r_2'}{}
\]
\vspace{6pt}
\[
\postrule{r_1 \derives{a} r_1'}{r_1 \cdot r_2 \derives{a} r_1' \cdot r_2}{}
\;\;\;\;
\postrule{r_1 \surd \;\;\;\; r_2 \derives{a} r_2'}{r_1 \cdot r_2 \derives{a} r_2'}{}
\;\;\;\;
\postrule{r \derives{a} r'}{r^* \derives{a} r' \cdot (r^*)}{}
\]

\end{document}